\documentclass{llncs}

\usepackage{algorithm2e}
\usepackage{amssymb}
\usepackage{amsmath}
\usepackage{graphicx}  

\newcommand{\Init}{\mathrm{Init}}
\newcommand{\Unsafe}{\mathrm{Unsafe}}

\newcommand{\comment}[1]{} 
\newcommand{\Ext}[1]{} 
\newcommand{\AL}[2]{#1} 

\newcommand{\Ceil}[1]{\lceil#1\rceil}
\newcommand{\Cont}[1]{\overline{#1}}

\newcommand{\RR}{\mathbb{R}}

\title{Combined Global and Local Search\\for the Falsification of Hybrid Systems
\thanks{This work was supported by the Czech Science Foundation (GA{\v C}R) grant number P202/12/J060 with institutional support RVO:67985807.}
}
\author{Jan Ku{\v r}{\'a}tko\inst{1}\inst{2} and Stefan Ratschan\inst{1}\thanks{ORCID: 0000-0003-1710-1513}}

\institute{Institute of Computer Science, Academy of Sciences of the Czech Republic \and Faculty of Mathematics and Physics, Charles University in Prague, Czech Republic}

\begin{document}

\maketitle

\begin{abstract}
In this paper we solve the problem of finding a trajectory that shows that a given hybrid dynamical system with deterministic evolution leaves a given set of states considered to be safe. The algorithm combines local with global search for achieving both efficiency and global convergence. In local search, it exploits derivatives for efficient computation. Unlike other methods for falsification of hybrid systems with deterministic evolution, we do not restrict our search to trajectories of a certain bounded length but search for error trajectories of arbitrary length.
\end{abstract}

\section{Introduction}

In this paper we provide an algorithm that solves the problem of unbounded safety falsification of hybrid systems with deterministic evolution. This means that, given a hybrid system with deterministic evolution, and a set of initial and a set of unsafe states, we search for a trajectory of arbitrary length starting in an initial state and ending in an unsafe state.

Existing methods for falsification of hybrid systems with deterministic evolution roughly fall into the following two categories:
\begin{itemize}
\item Local search~\cite{Abbas:11,Zutshi:13}: Such methods use local optimization to incrementally bring a starting trajectory closer to an error trajectory, ideally based on information on the derivative of the objective
  function. The advantage of local search is its relative efficiency. The disadvantage is that for convergence it needs to be started close enough to an error trajectory. At the very least it needs to start from a sequence of
  modes that contains an error trajectory. However, the number of sequences of modes grows exponentially with the length of the sequence which makes the search for starting trajectories for local search a difficult problem.
\item Black-box global search~\cite{Nghiem:10,Yashwant:11}: Such methods search for error trajectories globally, but use black-box optimization techniques~\cite{Gendreau:10,Rios:12} that do not explicitly exploit the structure specific to hybrid systems (partially continuous behavior, unbounded time variable). This extends their applicability (e.g., to Simulink models), but this may also result in loss of efficiency and restrict search to trajectories up to a given user-provided length. Of course it is possible to repeatedly restart such methods with higher upper bounds on the trajectory length, but every restart loses the information computed before. 
\end{itemize}

The contribution of this paper is an algorithm that combines the scalability of local search with global convergence for error trajectories of \emph{unbounded} length. Moreover, the resulting algorithm is reasonably simple and easy to analyze and implement. Note however, that efficiency is not primary goal of this paper--- since the generic structure of the resulting algorithm allows the simple incorporation more sophisticated global search techniques~\cite{Gendreau:10,Locatelli:13}.\AL{

Of course, one can use local search from any result of an algorithm based
on black-box global search, but this has the following drawbacks:}
{ Of course, one can use local search from any result of an algorithm based
on black-box global search, but this has the following drawbacks:}
\begin{itemize}
\item It is not clear how to handle the unbounded time variable.
\item Black-box global search does not explicitly exploit the structure of hybrid systems.
\item Black-box global search does not exploit the fact that it is combined with a local search method and hence may both duplicate some of the efforts of local search and fail to steer its search to good starting points for
  local search.
\end{itemize}

Our approach is based on a standard technique in global optimization for combining local with global search, so called two-phase methods~\cite{Schoen:02}. But we adapt those methods to the situation that we have here: A direct application of two-phase methods would use a  search space that is spanned by variables of two kinds:
the initial point of trajectories, and
the trajectory length (wrt. time). 
However, trajectory length is special, since it 
is unbounded, and since
computing a trajectory of the given length from a given initial point
  also computes all trajectories from that initial point with shorter length.
Moreover, hybrid systems combine continuous with discrete behavior and local search can exploit derivatives for searching the continuous part of the states space, but no such derivatives are available for discrete search which is another obstacle to the direct application of two-phase methods.

Hence, our approach modifies two-phase methods in such a way that---instead of treating trajectory length as a problem variable---they build trajectories incrementally from trajectory segments, and use derivative based continuous local search to glue together those segments based on continuous search (the literature on numerical algorithms for solving boundary value problems calls such an approach ``multiple shooting''~\cite{Ascher:95,Zutshi:13}).

The structure of the paper is as follows: In the next section we precisely define
the problem and introduce some basic definitions. In
Section~\ref{sec:algorithm} we introduce the main algorithm. In Section~\ref{sec:incrementality} we present an improved, more incremental version of the algorithm. In Section~\ref{sec:local-optimization} we describe how to do local search for
error trajectories. In Section~\ref{sec:terminator} we provide some termination proofs for the algorithm. In Section~\ref{sec:results} we present computational experiments. In Section~\ref{sec:related-work} we describe related
work, and in Section~\ref{sec:conclusion} we conclude the paper.

We thank Aditya Zutshi and Sriram Sankaranarayanan for interesting discussions on the
subject of this papers.

\section{Problem Formulation}
\label{sec:problem}

In this section we introduce notation and key concepts which we use, and present the problem we try to solve.
\begin{definition}
\label{def:hds}
A \textsf{hybrid dynamical system} is a quintuple $H = \left( Q, \Omega, F, G, R \right)$, where
\begin{itemize}
\item $Q$ is a finite set whose elements we call \emph{modes};
\item $\Omega \subseteq Q \times \RR^n$ (the \emph{state space} of the hybrid system);
\item $F$ assigns to each mode $q \in Q$ a system of differential equations $F_q(t, x, \dot{x}) = 0$, where $(q,x) \in \Omega$ and $t\in\RR^{\geq 0}$ is time;
\item $G \subseteq \Omega$ (the set of \emph{guards});
\item $R: \Omega \mapsto \Omega$ (the \emph{reset function}).
\end{itemize}
\end{definition}
For a given $q\in Q$, we will sometimes denote by $X_q$ the set $\{ x \mid (q, x)\in \Omega \}$.
\begin{definition}
\label{def:htra}
A \textsf{trajectory} of a hybrid dynamical system $H$ is a
sequence of the form $\left( (q_1, x_1), (q_2,
  x_2), \ldots, (q_k, x_k) \right)$, where $q_i \in Q$ and $x_i: [0,t_i]\mapsto
X_q$ is a continuous trajectory of the system of differential equations given by
$F_{q_i}$, $i = 1, \ldots, k$. For all $i\in \{ 1,\dots, k-1\}$, for all $t\in [0, t_i)$, not $G(q_i, x_i(t))$, but for the trajectory endpoints $G(q_i,x_i(t_i))$. Moreover, the starting points of subsequent trajectories are determined by the reset function, that is, $R\left((q_{i}, x_{i}(t_i)) \right)=(q_{i+1}, x_{i+1}(0))$.


We call $t_i\in\RR^{\geq 0}$ the \textsf{length} of $x_i$. Moreover, we denote by $(q_i, x_i^s) \in \Omega$ the starting point of a trajectory $(q_i, x_i)$ and $(q_i,x_i^e) \in \Omega$ its endpoint. \end{definition}
\Ext{
In this paper we only consider deterministic hybrid dynamical systems. By \emph{deterministic}, we mean that every trajectory $\left( (q_1, x_1), (q_2,
  x_2), \ldots (q_k, x_k) \right)$, is uniquely given by its initial state
$(q_1,x_1(0))$. Under this assumption, finding any trajectory of interest may be
reduced to finding its initial state $(q_1,x_1(0))$.} Now we are ready to formulate the problem of \emph{falsification of hybrid dynamical systems}.
\begin{problem}
\label{prob:fals}
Let $H$ be a hybrid dynamical system and $\Init \subset \Omega$, $\Unsafe \subset \Omega$ be two sets. The set $\Init$ is called the set of initial states and the set $\Unsafe$ is called the set of unsafe states. The problem of falsification of $H$ is to find any trajectory $\left( (q_1, x_1), (q_2,
  x_2), \ldots, (q_k, x_k) \right)$ of $H$ such that $(q_1,x_1^s) \in \Init$ and $(q_k,x_k^e) \in \Unsafe$. Such a trajectory is called an error trajectory of $H$.
\end{problem}
\Ext{There are two special cases. The first one is when $\Init \cap \Unsafe \neq \emptyset$. It may be difficult to check whether these sets intersect. However, when we find at least one point in $\Init \cap \Unsafe$ we are done. The second case, where there are no error trajectories in $H$, is more subtle because proving that a hybrid dynamical system features no error trajectories is undecidable \cite{Henzinger:98}.}

\section{Algorithm}
\label{sec:algorithm}

We now present the main algorithm. Throughout this section we will assume a given hybrid system $H$ with set of initial states $\Init$ and set of unsafe states $\Unsafe$.

Informally, we intend to transform Problem \ref{prob:fals} into the minimization of a cost function. A value of this cost function will measure how far or close a sequence of points in $\Omega$ is to an error trajectory. We will minimize this cost function until we find an error trajectory.


The algorithm will maintain a finite set of points $P\subseteq \Omega$ on which it will analyze the behavior of the given hybrid system $H$. We call sequences of elements from $P$ \emph{paths}. The algorithm will do this analysis by starting numerical simulations from points in $P$. Such a simulation will conclude that there is a trajectory from the starting point $p$ of a simulation to the endpoint $p'$. This will be stored in a relation $\rightarrow$ on $P$ that will relate all those points $p, p'$ in $P$ for which simulation showed that there is a trajectory from $p$ to $p'$ according to $H$. If there is a path from an initial point to an unsafe point according to the relation $\rightarrow$, we are done. However, since this is difficult to achieve, we allow paths of points in $P$ for which subsequent points are not in $\rightarrow$. In order to measure how far such a path is from being a trajectory we will now introduce a distance measure for points in $P$:


\begin{definition}
\label{def:dist_states}
Given a finite set of states $P\subseteq\Omega$ and a relation $\rightarrow\subseteq P\times P$, the
 \textsf{distance} $d((q,x), (q', x'))$, of states $(q,x)$ and $(q', x')$ in $P$, is
\begin{itemize}
\item $0$, if $(q,x)\rightarrow (q',x')$, otherwise
\item $\|x-x'\|$, if $q=q'$, and
\item $\infty$, otherwise.
\end{itemize}
\end{definition}
\AL{

Here, the symbol $\| \cdot \|$ denotes the Euclidean norm.

}
{Here, the symbol $\| \cdot \|$ denotes the Euclidean norm. }
Note that our distance function is not symmetrical because of the relation $\rightarrow$ that, in general, is not symmetrical which corresponds to the intuition that the existence of a trajectory from $p$ to $p'$ does not imply the existence of a trajectory from $p'$ to $p$.

We measure the difficulty of getting from an initial state of $H$ to a given state $(q,x)$, and from a given state $(q,x)$ to an unsafe state as follows:

\begin{definition}
\label{def:dist_sets}
For a state $(q, x) \in P$
\AL{
\begin{align*}
d_I\left((q, x)\right) & \equiv \inf_{u \in \Init} d\left( u, (q,x) \right), &
d_U\left((q, x)\right) & \equiv \inf_{u \in \Unsafe} d\left((q,x), u \right) .
\end{align*}
}
{ we put $d_I\left((q, x)\right) \equiv \inf_{u \in \Init} d\left( u, (q,x) \right)$ and $d_U\left((q, x)\right) \equiv \inf_{u \in \Unsafe} d\left((q,x), u \right)$.}
\end{definition}
\Ext{So, for a given state $(q,x) \in P$, if the mode $q$ does not contain any initial state and there is no relation $\rightarrow$ between any $u \in \Init$ and $(q,x)$, we put $d_I(q,x) = \infty$. Similarly, if the mode $q$ does not contain any unsafe state and there is no relation $\rightarrow$ between $(q,x)$ and any $u \in \Unsafe$, we put $d_U(q,x) = \infty$. }
\AL{

Now we model how close a path is to yielding an error trajectory, as follows:}
{Now we model how close a path is to yielding an error trajectory, as follows:}

\begin{definition}
\label{def:cost}
The \textsf{cost}
\AL
{ $c(p_1,\dots,p_n)$ of a path $(p_1,\dots,p_n)$ is 
\[
d_I(p_1) + \sum_{i=1}^{n-1} d(p_i,p_{i+1}) + d_U(p_{n}).
\]}
{ of a path $(p_1,\dots,p_n)$ is given by $c(p_1,\dots,p_n) = d_I(p_1) + \sum_{i=1}^{n-1} d(p_i,p_{i+1}) + d_U(p_{n})$.}
\end{definition}
Notice that we have an error trajectory of $H$ if the cost $c(p_1, \ldots, p_n)$
is equal to zero. \Ext{That means that there exists a trajectory $\left( (q_1, x_1), (q_2,
  x_2), \ldots (q_k, x_k) \right)$ with $(q_1, x_1^s) \in \Init$ and $(q_k, x_k^e) \in \Unsafe$ which passes through states $p_1, \ldots, p_n \in P$.} In practice, we are satisfied if the distances $d_I(p_1)$ and $d_U(p_{n})$ are zero and $\sum_{i=1}^{n-1} d(p_i,p_{i+1})<\varepsilon$ for some small threshold $\varepsilon$. 




Now we can formulate our method for falsification of hybrid dynamical systems.
In a similar way as two-phase methods~\cite{Schoen:02} the algorithm iterates between two phases for exploring the state space of a given hybrid dynamical system.

The first phase is local optimization. For a given path of finite cost from $P$ we compute another path with lower cost. In this phase we employ standard techniques for continuous optimization and use gradient information based on sensitivity analysis of hybrid dynamical systems~\cite{Hiskens:00}. If we find a local minimum which yields an error trajectory, we are finished. However, if the minimal cost is greater than a given threshold $\varepsilon$, we need to proceed to phase two and explore the state space further. The reader can find more details on the first phase in Section~\ref{sec:local-optimization}.

The second phase is called global exploration. If local optimization in the first phase does not produce an error trajectory, we add additional states to the set $P$. There are many options for adding new states such as a random sampling, states resulting from forward and backward simulation from existing states, states suggested by more sophisticated global search techniques~\cite{Gendreau:10,Locatelli:13}, and even states given by a designer of the system. 
\AL{

The complete Algorithm \ref{algo:fals} follows:}
{ The complete Algorithm \ref{algo:fals} follows:}

%
\begin{algorithm}[H]
\label{alg:1}
\SetAlgoNoLine
\DontPrintSemicolon
\KwIn{a set of states $P\subseteq\Omega$ and a relation $\rightarrow\subseteq P\times P$ s.t.
  \begin{itemize}
  \item $p\rightarrow p'$ implies that there is a trajectory from $p$ to $p'$ in $H$
  \item there is a path of points in $P$ that has finite cost with respect to $\rightarrow$
  \end{itemize}}
\KwOut{an error trajectory} 
\vspace*{0.5cm}
\While{local optimization of the path with minimal cost does not yield an error trajectory}{
add a new state $r \in \Omega$ to the set $P$\;
\For{some $p \in P$}{
simulate forward from $p$ for some time to a new state $p'$\; $\rightarrow\; := \ \rightarrow \cup \{(p, p') \}$\; 
}
\For{some $p \in P$}{
simulate backward from $p$ for some time to a new state $p'$\; $\rightarrow\; := \ \rightarrow \cup \{(p', p) \}$\;
}
}
\caption{\AL{Combined Global and Local Search for the Falsification of Hybrid Systems}
{ Combined Global and Local Search for the Falsification}}
\label{algo:fals}
\end{algorithm}
\AL{}{\noindent}
The second requirement on the input (existence of a path of finite cost) allows us to use derivative based continuous local optimization from the beginning. For fulfilling this requirement, we observe that paths can only have infinite cost due to sub-sequent states in different modes that are not connected by the relation $\rightarrow$. We can make this more precise by the following property:

\begin{property}
  Assume a set $P\subseteq \Omega$ and $\rightarrow\subseteq P\times P$. Let $\rightarrow_Q\subseteq P\times P$ be such that $(q, x)\rightarrow_Q (q',x')$ iff $q=q'$ or $(q,x)\rightarrow (q',x')$, and let $\rightarrow_Q^*$ be the transitive closure of $\rightarrow_Q$. If $P$ contains at least one initial state $p$ and one unsafe state $p'$ such that $p\rightarrow_Q^*p'$, then there is a path of points in $P$ that has finite cost with respect to $\rightarrow$
\end{property}

The necessary elements of $\rightarrow$ can be easily formed by pairs $(p, p')$ such that 
 $G(p)$ and $p'=R(p)$. For example, for each pair of modes $(q, q')$ for which there are $x$ and $x'$ s.t. $G(q,x)$ and $(q',x')=R(q,x)$ we could add such $(q,x)$ and $(q',x')$. If $H$ has an error trajectory then this fulfills the assumptions of the above property resulting in a path of finite cost.

The algorithm stops when we find a path whose cost is lower than some threshold $\varepsilon$. A concrete implementation might add another stopping criterion, for example stating a maximum number of states in $P$ in order to ensure termination for inputs that do not feature any error trajectory. 



\section{Algorithmic Details}
\label{sec:incrementality}

\subsection{Computation of the Path of Minimal Cost}
We make the following observation: In Algorithm~\ref{alg:1}, one can view the problem of computation of a
path of minimal cost as a problem on weighted directed graphs: The vertices of the graph are formed by the elements of the set $P$ and there is an edge from $p\in P$ to $p'\in P$ iff the distance $d(p,p')$ is finite. The weight of this edge is given by this value $d(p,p')$. Now the path of minimal cost is the shortest path in this graph from an initial to an unsafe state. This is a classical problem in algorithm theory with solutions such as the Floyd-–Warshall algorithm. 

Examining the situation more closely, we observe that our problem is neither of
the all-pair shortest path, nor of the single-source shortest path
kind. Instead, the paths have to start in a certain given set (call it $S$ for
source) and end in another given set (call it $G$ for goal). This can be reduced
to a problem with single vertices instead of sets by introducing two new,
auxiliary vertices $s$ and $g$ such that $s$ has an edge of zero cost to each
element of $S$ and such that there is an edge of zero cost from each element of
$G$ to $g$. 

Now we are left with a single-source single-goal shortest path problem (also called point-to-point shortest path). Of course, such problems can be solved by algorithms solving the single-source shortest path problem, for example, by Dijkstra's algorithm~\cite{Dijkstra:59}. But there are also specialized algorithms, for example, algorithm based on a bi-directional~\cite{Pohl:71,Bertsekas:98} application of Dijkstra's algorithm.

\subsection{Heuristics}

The algorithm can be instantiated with many heuristics resulting in special versions of Algorithm~\ref{alg:1}, for example:
\begin{itemize}
\item Forward version: only add initial points and only do forward simulation
\item Backward version: only add unsafe points and only do backward simulation
\item Complete random search version: never prolong existing simulations, only simulation from newly added points. 
\end{itemize}

The algorithm also leaves open the length of the employed simulations. A simply possibility is to fix a certain length at the beginning and stick to it throughout computation\comment{\footnote{more details: advantage of short vs. long simulations, smarter approaches, e.g. cancel simulation in points where distance of guard of unsafe set is minimal}}. 

Note that simulation might run into problems, for example due to Zeno behavior, or due to the fact that it leaves the state space of the given hybrid system. In this case we simply ignore the result of simulation and continue with the algorithm. See also more on this at the end of Section~\ref{sec:local-optimization}.

It is also possible to use information from verification tools here. Especially, one can restrict the choice of points to an abstraction computed by a verification tool~\cite{Dzetkulic:11,Ratschan:09a}.

\subsection{Paths of Minimal Cost}

We will now investigate the form of paths of minimal cost. For a given hybrid system $H$ we assume the following.
\begin{enumerate}
\item The sets $\{ x \mid (q, x) \in \Init \}$ and $\{  x \mid (q, x) \in \Unsafe \}$ are closed and convex.
\item For all $q \in Q$ the set $\{ x \mid (q, x) \in \Omega \}$ is convex.
\item For $p, p' \in P$, $p \rightarrow p'$ implies there is a trajectory from $p$ to $p'$ in $H$.
\item There is at least one path of finite cost in $P$ with respect to $\rightarrow$\;.
\end{enumerate}
\begin{lemma}
\label{lem:SSiM}
Let $H$ be a given hybrid system and $P$ be a set of states such that assumptions 1.--4 hold. Let $(p_1, \ldots, p_n)$ be the path of minimal cost. Let $r$ be a state such that $r \neq p_i$, $i = 1, \ldots, n$, and neither $r \rightarrow p_i$ nor $p_i \rightarrow r$ for any $i\in \{ 1,\dots,n\}$. Then the cost of a path which is formed by either including state $r$ in $(p_1, \ldots, p_n)$ or substituting $r$ for any $p_i$'s, $i \in \{ 1, \ldots, n\} $, in $(p_1, \ldots, p_n)$ is greater or equal to $c(p_1, \ldots, p_n)$.
\end{lemma}
The reader can find the proof \AL{of Lemma \ref{lem:SSiM} in the appendix}{of Lemma \ref{lem:SSiM} in the extended version of the paper}. The consequence of this lemma can be stated like this: Assuming 1.--4. the value of $c(p_1, \ldots, p_n)$, where $(p_1, \ldots, p_n)$ is the path of minimal cost, does not depend on states $p \in P$ which are in no relation to other states in the path wrt. $\rightarrow$. In other words, for a state $p_i$, $i = 1, \ldots, n$, which is in no relation with other states in a path, we have $c(p_1, \ldots, p_i, \ldots, p_n) = c(p_1, \ldots, p_{i-1}, p_{i+1}, \ldots, p_n)$.

This is important for finding paths of lower cost. Whenever we add a new state $r$ in Algorithm \ref{algo:fals} we should simulate either forward in time or backward in time to create a pair of states in a relation $\rightarrow$. Solitary states do not affect the resulting value of the cost function.
\begin{lemma}
\label{lem:3SoT}
Let $H$ be a given hybrid system and $P$ be a set of states such that assumptions 1.--4 hold. Let $(p_1, \ldots, p_n)$, $p_i \in P$, be the path of minimal cost. Let $p_j$, $j = 2, \ldots, n-1$, be a state such that $p_{j-1} \rightarrow p_j$ and $p_j \rightarrow p_{j+1}$. Then $d(p_{j-1}, p_{j+1}) = 0$. \qed
\end{lemma}
\begin{proof}
Due to transitivity of the relation $\rightarrow$, we have $p_{j-1} \rightarrow p_{j+1}$ which gives us $d(p_{j-1}, p_{j+1}) = 0$.
\end{proof}
Lemma~\ref{lem:3SoT} presents us with a choice for the application of local optimization. If a path contains such a triplet $(p_{j-1}, p_j, p_{j+1})$, we may either work with them as two separate hybrid trajectories or we may consider a hybrid trajectory which is formed by their connection, thus removing the intermediate state from a path.  On the other hand, we also have the option to split a hybrid trajectory into shorter hybrid trajectories before passing it to the local optimizer.


In the first approach we work with shorter trajectories however the resulting optimization problem has higher dimension. The latter case may, on the other hand, cause problems because it is less numerical stable~\cite{Ascher:95}. The choice depends on the system of differential equations that governs the evolution of hybrid system $H$.

\section{Local Optimization}
\label{sec:local-optimization}
In Algorithm \ref{algo:fals}, after we form a path $(p_1, \ldots, p_n)$ of finite cost  we  try to find another path of smaller cost using local search. Therefore, we solve the minimization problem in which we seek new states $\hat{p}_1, \ldots, \hat{p}_n$ which yield a path of lower cost than $(p_1, \ldots, p_n)$. Eventually, such a path of minimal cost may correspond to an error trajectory of a hybrid system.

The efficiency of such local search can be improved by exploiting the gradient of the cost $c(p_1, \ldots, p_n)$. In this section  we will develop explicit formulae for the gradient of the cost function which will allow us to use efficient off-the-shelf tools for gradient-based numerical optimization to minimize the cost function.

Without loss of generality, we will assume that for all $i\in\{1,\dots,n\}$, $p_i\rightarrow p_{i+1}$ iff $i$ is odd. This can be easily achieved, since, due to Lemma~\ref{lem:SSiM} if there are solitary states that we can remove them from $(p_1, \ldots, p_n)$ without changing the value of the cost function. \Ext{Moreover, for states $p_{j-1}$, $p_j$, $p_{j+1}$ with $p_{j-1}\rightarrow p_j\rightarrow p_{j+1}$, we can either duplicate $p_j$ or remove $p_j$, depending whether we prefer longer or shorter trajectories in the sub-sequent local optimization process.}

Note that in contrast to early work~\cite{Zutshi:13}, in cases where $p_i\rightarrow p_{i+1}$, $p_i$ and $p_{i+1}$ are \emph{not} restricted to be in the same mode. Moreover, points are not restricted to reside in guards of the hybrid system.



We now give explicit formulae for computation of the gradient of the cost function $c(p_1, \ldots, p_n)$. Let us start with the definition of the length of a trajectory $\left( (q_1,x_1), \ldots, (q_k, x_k) \right)$ and its sensitivity to the change of its initial state $(q_1,x_1^s)$ which is essential for evaluation of the gradient of the cost $c(p_1, \ldots, p_n)$.

\begin{definition}
\label{def:length}
The \textsf{length} of a trajectory $\left( (q_1,x_1), \ldots, (q_k, x_k) \right)$ is defined to be the sum $t_f  =  \sum_{i=1}^kt_i$, where $t_i$ is the length of $x_i$ for  $i = 1, \ldots, k$.
\end{definition}
\Ext{Such a definition is natural because each segment $(q_i, x_i)$ of a trajectory is given by the continuous trajectory $x_i$ of the system of differential equations $F(q_i)$ on a time interval $[0,t_i]$. Therefore, in order to get from the state $(q_1,x_1^s)$ to the state $(q_k, x_k^e)$ we need $t_f$ units of time.}
\begin{definition}
\label{def:M}
We define a function $M: \RR \times \Omega \mapsto \Omega$ such that for a state $(q,x)\in\Omega$
and $t \in [0, t_f]$ we have 
\AL{
\[
M(t, (q, x)) =(q',x')\  ,
\]}
{$M(t, (q, x)) =(q',x')$, }
where $(q',x')$ is the end-state of the trajectory of length $t$ whose initial state is $(q,x)$. In cases where a reset happens at time $t$ (which results in the trajectory of length $t$ being non-unique), we choose the unique point before the reset.
\end{definition}
\begin{definition}
\label{def:sen}
The \textsf{sensitivity} of a trajectory $\left( (q_1,x_1), \ldots, (q_k, x_k) \right)$ of $H$ to the initial state $(q_1, x_1^s)$ is a function $S: \RR^{\geq 0} \mapsto \RR^{n \times n}$ such that
\[
S(t) \equiv \frac{\partial M(t, (q, x_1^s))}{\partial x_1^s}, \quad t \in [0, t_f].
\]
\end{definition}
With this sensitivity function we can measure how the states on a hybrid trajectory are affected when we change its initial state. An important observation is that $S(0)$ is the \emph{identity} matrix. However, for hybrid systems, the function $M$ need not be differentiable everywhere, and so the sensitivity is not defined everywhere. Computation of the sensitivity function is subtle~\cite{Hiskens:00}. \Ext{The main obstacle is updating the sensitivity function at discrete events. When we leave one mode and enter another, we need to update the sensitivity function according to formula (57) in~\cite{Hiskens:00}. Obtaining the sensitivity function to an initial state in the classical continuous case, that is, when we consider only one mode, is described in~\cite[p. 99--102]{Khalil:02}.} 
\AL{

For our path $(p_1,\dots,p_n)$, for certain $t_i$,  $i = 1, 3, \ldots, n-1$, we have
\[
M(t_1, p_1)  = p_2, \quad
M(t_3, p_3) = p_4, \quad 
\ldots, \quad
M(t_{n-1}, p_{n-1})  = p_n\ .
\]
In the sequel let us use the following notation: For any state $(q, x)\in\Omega$, we denote by $\Cont{(q, x)}$ its continuous part $x$. 

}
{ In the sequel let us use the following notation: For any state $(q, x)\in\Omega$, we denote by $\Cont{(q, x)}$ its continuous part $x$. 

For our path $(p_1,\dots,p_n)$, for certain $t_i$,  $i = 1, 3, \ldots, n-1$, we have
$
M(t_1, p_1)  = p_2,\  
M(t_3, p_3) = p_4, \ 
\ldots, \ 
M(t_{n-1}, p_{n-1})  = p_n$. }
Local search adjusts the position of the initial state of each trajectory together with its length such that the cost is minimized. It uses the gradient of the cost with respect to $\Cont{p_i}$ and lengths $t_i$ for $i < n$ odd. Therefore the gradient is given by the partial derivatives $\frac{\partial c}{\partial \Cont{p_i}}(p_1, \ldots, p_n)$ and $\frac{\partial c}{\partial t_i}(p_1, \ldots, p_n)$, $i< n$ odd. 

We will illustrate the whole process of computing the gradient of the cost function for one particular definition of distances $d_I, d(\cdot, \cdot)$ and $d_U$ that avoids solving another minimization problem stemming from Definition~\ref{def:dist_sets}. Hence, we put $d_I(p_1)$ and $d_U(p_n)$ to be weighted norms to some fixed states in $\Init$, and $\Unsafe$ respectively.
This amounts to the sets $\Init$ and $\Unsafe$ being ellipsoids.
 We denote by $u\in\Omega$ and $v\in\Omega$ the centres of these ellipsoids and by $E_I, E_U$ symmetric positive definite matrices which characterize the size and shape of sets $\Init$ and $\Unsafe$. \AL{Then we consider the cost of the following special form
\begin{align*}
c(p_1, \ldots, p_n) & = d_I(p_1) + \sum_{i=1}^{n-2}d(p_i, p_{i+1}) + d_U(p_n)\ ,\\
 & = \| \Cont{p_1} - \Cont{u}\|_{E_I}^2 + \sum_{i \ \mathrm{even}}^{n-2}\| \Cont{p_i} - \Cont{p_{i+1}} \|^2 + \|\Cont{p_n} - \Cont{v}\|_{E_U}^2\ .
\end{align*}}
{

Then we consider the cost of the following special form: $c(p_1, \ldots, p_n)  = d_I(p_1) + \sum_{i=1}^{n-2}d(p_i, p_{i+1}) + d_U(p_n)  = \| \Cont{p_1} - \Cont{u}\|_{E_I}^2 + \sum_{i \ \mathrm{even}}^{n-2}\| \Cont{p_i} - \Cont{p_{i+1}} \|^2 + \|\Cont{p_n} - \Cont{v}\|_{E_U}^2$. }
When we use the function $M: \RR \times \Omega \mapsto \Omega$ from Definition \ref{def:M}, then the cost becomes dependent on $p_i$ and $t_i$, $i = 1, 3, 5, \ldots, n-1$,
\AL{
\begin{align*}
 c(p_1, p_3, \ldots, p_{n-1}, t_1, t_3, \ldots, t_{n-1}) & = \| \Cont{p_1} - \Cont{u}\|_{E_I}^2 + \ldots \\
 &  + \sum_{i \ \mathrm{even}}^{n-2}\| \Cont{M(t_{i-1},p_{i-1})} - \Cont{p_{i+1}} \|^2 + \ldots \\
 & +  \|\Cont{M(t_{n-1},p_{n-1})} - \Cont{v}\|_{E_U}^2 \ .
\end{align*}}
{ and $c(p_1, p_3, \ldots, p_{n-1}, t_1, t_3, \ldots, t_{n-1})  = \| \Cont{p_1} - \Cont{u}\|_{E_I}^2 +  \sum_{i \ \mathrm{even}}^{n-2}\| \Cont{M(t_{i-1},p_{i-1})} - \Cont{p_{i+1}} \|^2 +  \|\Cont{M(t_{n-1},p_{n-1})} - \Cont{v}\|_{E_U}^2$.

}
\AL{We can compute the gradient of the cost $ c(p_1, p_3, \ldots, p_{n-1}, t_1, t_3, \ldots, t_{n-1})$ which consists of the following partial derivatives
\begin{align*}
\frac{\partial c}{\partial \Cont{p_1}} & = 2\left[(\Cont{p_1} - \Cont{u})^TE_I + \left(\Cont{M(t_1, p_1)} - \Cont{p_3}\right)^T\frac{\partial M}{\partial \Cont{p_1}}(t_1, p_1)\right] 
\end{align*}
and for odd $i$ with $1 < i < n-1$
\begin{align*}
\frac{\partial c}{\partial \Cont{p_i}} & =-2\left[\left(\Cont{M(t_{i-2}, p_{i-2})} - \Cont{p_i}\right)^T - \left(\Cont{M(t_i, p_i)} - \Cont{p_{i+2}}\right)^T\frac{\partial M}{\partial \Cont{p_i}}(t_i, p_i)\right]
\end{align*}
with
\begin{align*}
\frac{\partial c}{\partial \Cont{p_{n-1}}} & = -2\left[\left(\Cont{M(t_{n-2}, p_{n-2})} - \Cont{p_{n-1}}\right)^T - \right. \\
& - \left.  \left(\Cont{M(t_{n-1},p_{n-1})} - \Cont{v}\right)^TE_U\;\frac{\partial M}{\partial \Cont{p}_{n-1}}(t_{n-1}, p_{n-1})\right] \ .
\end{align*}
For odd $i < n-1$ we put
\begin{align*}
 \frac{\partial c}{\partial {t_i}} & = 2(\Cont{M(t_{i}, p_{i})}-\Cont{p_{i+1}})^T\frac{\partial M}{\partial t_i}(t_i, p_i)\ 
 \end{align*}
 and
 \begin{align*}
\frac{\partial c}{\partial {t_{n-1}}} & = 2(\Cont{M(t_{n-1}, p_{n-1})}-\Cont{v})^TE_U\frac{\partial M}{\partial t_{n-1}}(t_{n-1}, p_{n-1})\ .
\end{align*}}
{We can compute the gradient of the cost $ c(p_1, p_3, \ldots, p_{n-1}, t_1, t_3, \ldots, t_{n-1})$ which consists of the following partial derivatives
\[
\frac{\partial c}{\partial \Cont{p_1}}  = 2[\Cont{p_1} - \Cont{u}]^TE_I +2 \left[\Cont{M(t_1, p_1)} - \Cont{p_3}\right]^T\frac{\partial M}{\partial \Cont{p_1}}(t_1, p_1)
\]
and for odd $i$ with $1 < i < n-1$ we have
\[
\frac{\partial c}{\partial \Cont{p_i}}  = - 2\left[\Cont{M(t_{i-2}, p_{i-2})} - \Cont{p_i}\right]^T +  2\left[\Cont{M(t_i, p_i)} - \Cont{p_{i+2}}\right]^T\frac{\partial M}{\partial \Cont{p_i}}(t_i, p_i)
\]
with the last term
\begin{align*}
\frac{\partial c}{\partial \Cont{p_{n-1}}}  & = 2\left[\Cont{M(t_{n-1},p_{n-1})} - \Cont{v}\right]^T\!E_U\frac{\partial M}{\partial \Cont{p}_{n-1}}(t_{n-1}, p_{n-1}) \\ 
& -2\left[\Cont{M(t_{n-2}, p_{n-2})} - \Cont{p_{n-1}}\right]^T\ .
\end{align*}
For odd $i < n-1$ we put 
\[
\frac{\partial c}{\partial t_i}  = 2[\Cont{M(t_{i}, p_{i})}-\Cont{p_{i+1}}]^T\frac{\partial M}{\partial t_i}(t_i, p_i)
\]
 and the last term is
 \[
 \frac{\partial c}{\partial {t_{n-1}}}  = 2\left[\Cont{M(t_{n-1}, p_{n-1})}-\Cont{v}\right]^TE_U\frac{\partial M}{\partial t_{n-1}}(t_{n-1}, p_{n-1})\ .
 \] }
 In addition we may introduce weights into the cost function to scale the problem. 
\Ext{Notice our redefined distances $d_I$ and $d_U$ function as penalty terms so we do not need to solve the minimization problem in Definition \ref{def:dist_sets}.
Using a penalty term for the set of initial and unsafe states is a common
technique in numerical optimization~\cite[Section~17]{Nocedal:99}. Such penalty
terms may in some situations change the local optima of optimization problems
(so-called ''inexact penalty terms''). However, in situations like ours, where
simplicity and speed are more important than finding an exact optimum, such
methods are frequently used. }

\AL{Now we can use numerical optimization algorithms with the cost function $c$ and the gradients defined as above, to do local search for paths of minimal cost.
If started close enough to an error trajectory, and if the hybrid system is
sufficiently well-behaved around the error trajectory, such local search will converge
(usually quickly). However, if this is not the case, local search may fail, due to various reasons:}
{Now we can use numerical optimization algorithms with the cost function $c$ and its gradient to do local search for paths of minimal cost.
If started close enough to an error trajectory, and if the hybrid system is
sufficiently well-behaved around the error trajectory, such local search will converge
(usually quickly). However, if this is not the case, local search may fail, due to various reasons:}
\begin{itemize}
\item It may run in a local minimum that is not an error trajectory.
\item There may be problems due to the fact that the sensitivity is not
  everywhere continuously differentiable. This corresponds to the situation
  where a trajectory is tangential to the boundary of a guard~\cite{Hiskens:00}.
\item Well-known problems with simulation~\cite{Mostermann:99} of hybrid
  systems might arise. For example, the simulation might run into Zeno behavior,
  or the ODE solver is unable to start close to the boundary of a guard. 
\item Optimization may result in trajectories that leave the state space of the
  hybrid system. 
\end{itemize}

In all such cases, we simply terminate local optimization and continue with the
global phase of the main algorithm.

\comment{\footnote{We might also say something about convergence speed.}}

\section{Termination Proof}
\label{sec:terminator}

We assume a hybrid system $H$ with the following properties:

\begin{itemize}
\item The state space of $H$ is compact.
\item There exists an error trajectory $E$ with final point in the interior of the set of unsafe points, and an $\varepsilon>0$ such that starting local search from any sequence of hybrid trajectories with cost not bigger than $\varepsilon$ converges to an error trajectory.
\item There exists a tube around $E$ such that trajectories starting in this tube depend continuously on their initial value (note that for ODEs this can be ensured by Lipschitz continuous right-hand sides).
\end{itemize}

Moreover, we will study a variant of the algorithm with the following properties:

\begin{itemize}
\item The algorithm does at least one forward simulation in each cycle, always of length (in time) $T$.
\item The algorithm chooses the starting point for its simulations randomly using a distribution that is non-zero on the whole state space.
\item If a simulation hits an unsafe state, it finishes (so, in such cases, the length of the simulation may be shorter than $T$).
\item The simulations are exact, that is, we ignore rounding and discretization errors of ODE solvers.
\end{itemize}

Note that the assumptions are asymmetrical wrt. time and set of initial vs. unsafe states. This is necessary since simulations have to be done in a certain direction, and since convergence requires simulations to be stopped if reaching an unsafe state.

While it is obvious that such an algorithm will densely fill the state space of the hybrid system with initial values of simulation, it is not obvious that this will eventually result in a path of small enough cost, since---from a given initial value---the trajectories follow the dynamics of the hybrid system $H$. Still, we have:

\begin{theorem}
\label{thm:1}
Under the assumption above, the algorithm finds an  error trajectory with probability $1$.
\end{theorem}
The reader can find \AL{the proof in the appendix}{proofs of Theorems \ref{thm:1} and \ref{thm:2} in the extended version of the paper}. Clearly one can easily get a dual version of the theorem and proof by turning around the time axis, switching initial and unsafe states etc. 

Only slightly changing the proof, one can prove the following non-probabilistic version of the theorem:

\begin{theorem}
\label{thm:2}
  Take the same assumptions as the previous theorem with the exception of choice
  of starting points of simulations. Instead of a choice according to some probability distribution, assume a choice of those starting points that fulfills the following property: For each $\varepsilon>0$, there is an integer $k$ such
  that for every $\varepsilon$-ball with center in the state space contains a
  simulation starting point that the algorithm has chosen in the first $k$
  iterations. Then algorithm always finds an error trajectory.
\end{theorem}

\comment{
Convergence in linear case with fixed endpoint.

mentioned which measure we use

refer to multi-start literature 

stronger condition for canceling simulations?}

\section{Computational Experiments}
\label{sec:results}

Recall that one of the main goals of our method was to handle the absence of an a-priori upper bound on the length of error trajectories. In order to study the cost of having to work without this information we compare our approach (that we will call ``unbounded method'') with another approach that also combines global with local search, but that does simulations of fixed length (we will call it ``bounded method''). The bounded method will also use derivative-based local optimization, but for initializing local optimization it randomly generates initial states in the mode containing $I$ and simulates for the time interval $[0, T]$. Whenever the resulting trajectory reaches the mode containing the set of unsafe states $U$, we take it as a starting trajectory for local search for an error trajectory. If we obtain an error trajectory then we stop. Otherwise we proceed until we generate a certain number of trajectories (we will denote this number by $M$). 

Note that any method that inspects the given hybrid system only up to a fixed time bound $T$, if $T$ is too small, it will not find any error trajectory at all. The bounded method that we use here, for $T$ too small, may never reach a mode containing $U$, preventing it from finding any error trajectory. Moreover, examples for which trajectories leading to the mode containing $U$ lead over a very small guard, become arbitrarily difficult for the bounded method. So we can already conclude now---without running any experiments---that the unbounded method is superior in certain cases.

Still we do some experiments with a widely known benchmark, the Navigation benchmark with 16 modes~\cite{Fehnker:04}. We consider the linear dynamics $\dot{x} = Ax - Bu(i, j)$, with $A$ and $B$ as usual for the navigation benchmark, and 
\begin{align*}
u(i, j) & = 
\begin{bmatrix}
\sin\left( \frac{\pi\: C(i, j)}{4} \right) \\
\cos\left( \frac{\pi\: C(i, j)}{4} \right)
\end{bmatrix} \ ,&
C & =
\begin{bmatrix}
4 & 3 & 3  & 4 \\
4 & 4 & 4 & 4 \\
4 & 6 & 6 & 4 \\
1 & 0 & 7 & 6
\end{bmatrix}
\ .
\end{align*}
Assume the sets of initial and unsafe states to be \emph{ellipsoids} such that their principal axes have length $0.2$, $0.2$, $2$ and $2$, however, $I$ is centred at $\begin{bmatrix}
0.5 & 3.5 & 0 & 0
\end{bmatrix}^T$ and $U$ is centred at $\begin{bmatrix}
3.5 & 1.5 & 0 & 0
\end{bmatrix}^T$. Our objective is to find any trajectory which starts in set $I$ and reaches set $U$.

For our experiments we use an instantiation of the unbounded method that fulfills the requirement of starting with a set $P$ that has a path of finite cost as follows: We initialize the set $P$  by putting a point on each boundary of two neighboring modes, simulating forward and backward from each such point (0.05 time units in each direction), and adding the endpoints to the set $P$. 

In the main algorithm, we add a random state to each mode (with velocities $x_3$ and $x_4$ ranging from $-1$ to $1$) and then we simulate forward and backward in time from such a state (0.5 time units in each direction). The extremities of the resulting error trajectory (its initial and end states) are stored in $P$ and used for obtaining a path of the minimal cost for local search. If local search returns an error trajectory, then we stop. As in the bounded method  we restrict computation, but this time, to add up to $M$ states to the set $P$. 

In all our experiments we do local search using the Scilab function for gradient-based numerical optimization, computing the gradient as described in Section~\ref{sec:local-optimization}. We weighted each distance between two consecutive segments by the weight $\omega = 500$ in order to prefer the continuity of resulting trajectories. We use the Scilab function \emph{rand} for generating random states. For reducing dependence of the result on the random number generator, we always carry out $100$ experiments, initializing the random number generator with a different seed (concretely, \emph{rand(``seed'', $i$)}, where $i = 1, \ldots 100$). In all our experiments we use the value $500$ for the constant $M$. The results are listed in Tab. \ref{Tab:comp1}.
The column ``successful falsification'' lists the number of experiments (from $100$) for which the method found an error trajectory. The column ''average total simulation time'' is the average of the length of all simulation done during a given experiment, but \emph{only} for those experiments that succeeded in finding an error trajectory.

\begin{table}
\begin{tabular}{l|c|c}
 & \ successful falsification\ & \ average total simulation time \\
 \hline 
 unbounded method & $99$ & $1937$ \\
 bounded method, $T=10$ & $85$ & $1260$ \\
 bounded method, $T=20$ & $89$ & $2935$
\end{tabular}
\caption{Computation Results}
\label{Tab:comp1}
\end{table}

The choices $T=10$ and $T=20$ that we used are big enough, so the bounded method does find error trajectories, but still the success rate is lower than with the unbounded method, that does not need any bound $T$ at all. In those cases where the bounded method actually finds an error trajectory, if the choice of $T$ is small but large enough to reach a mode containing $U$, it needs less simulation than the unbounded method. But very quickly, when not choosing $T$ small enough, also the cost of simulations increases beyond the unbounded method.

\Ext{
 \begin{figure}
 \centering
 \includegraphics[scale=1.0]{drift}
 \caption{A projection onto $x_1$--$x_2$ plane. We mark nodes starting from top-left, $(i, j) = (1, 1)$, to bottom-right which is $(i, j) = (4, 4)$. Arrows show direction of each drift.}
 \label{Fig:drift}
 \end{figure}
}

To sum up, the unbounded method significantly increases the chance of finding an error trajectory, and moreover, it also decreases the amount of simulation needed for that, except for cases where a very good bound on the error trajectory length is available.

\comment{starting from shortest path versus starting from all paths?

In multiple shooting we currently do not take the closest point to the set of unsafe states, but the final point of the simulation. It might make sense to change that.

Honza cancels simulations if they cross modes. This is not described currently.
}

\section{Related Work}
\label{sec:related-work}

Our algorithm can be viewed as an adaptation of the Best Start two-phase method for global optimization~\cite{Schoen:02} to our context. 

The falsification problem can also be viewed as a boundary value problem which is a classical topic in numerical mathematics~\cite{Ascher:95}. However, classical numerical methods assume a fixed final time, whereas we search for error trajectories of arbitrary length. Moreover, classical methods for boundary values problems are restricted to purely continuous systems and the formulation of boundary conditions as equalities. \Ext{Note that, unlike the continuous case, the trajectory of a hybrid system can reach a point $x(t)$ for which $f(x(t))\leq 0$ holds \emph{without} every reaching a point $x(t')$ for which $f(x(t))=0$ holds.}

Zuthsi and co-authors~\cite{Zutshi:13} present a method for falsification of hybrid systems that also uses multiple shooting based local search. However, the method assumes a given upper bound on the length of the error trajectory the method searches for. Moreover, their local search method always follows a 
 given sequence of modes and transitions. They propose to search for such a sequence using tools that compute abstractions of hybrid systems, or by random search. The form of the used trajectory segments is more restricted than in our method since trajectory segments always stay in one mode, and end in the guard leading to another mode.

Abbas~\cite{Abbas:11} show how to use local search for falsification of hybrid systems with affine dynamics. They propose to start the method from the result of global search algorithms~\cite{Nghiem:10}.

The usage of abstractions for guiding local search for error trajectories has been proposed earlier~\cite{Ratschan:09a}, in combination with the usage of derivative-free algorithms for local search. 

There is more related work for systems that---different from our case---allow input or have non-deterministic dynamics. In the completely discrete case this amounts to 
 finding shortest paths in graphs~\cite{Bertsekas:98}. We use shortest path algorithms as a sub-algorithm to find starting points for local search. Similar problems are studied in more structured domains by the field of planning~\cite{LaValle:06}, and in formal verification by directed model checking~\cite{Edelkamp:09}.

In the continuous case, the classical field studying algorithm for finding paths of dynamical systems that are in some sense optimal (e.g., as short as possible), is optimal control~\cite{Betts:98,Branicky:98a}. In recent years, also the field of planning has started to study continuous dynamical systems~\cite[Chapter IV: Planning Under Differential Constraints]{LaValle:06} from a different perspective. More recently, such techniques have also been applied to hybrid systems~\cite{Branicky:06,Dang:09,Plaku:09,Plaku:09a}. Planning-based techniques search globally, and do not require an upper bound on trajectory length, but they do not incorporate derivative-based local search. The only exception that we are aware of~\cite{Lamiraux:04} uses optimal control to the  result of planning in a purely sequential way, without any iteration between the two phases.




\Ext{Hamilton-Jacobi type methods, application of shortest path algorithm to Hamilton-Jacobi discretization~\cite{Tsitsiklis:95}}


\section{Conclusion}
\label{sec:conclusion}

We presented an algorithm for the falsification of hybrid system that combines
scalability due to local search with convergence due to global search. In future
work, we will improve the algorithm in analogy to advanced two-phase
methods~\cite{Schoen:02}, such as clustering methods that exploit the
regions of attraction to local optima of the used local search technique.

\bibliographystyle{abbrv}
\bibliography{sratscha}

\newcommand{\SortNoop}[1]{}
\begin{thebibliography}{10}

\bibitem{Abbas:11}
H.~Abbas and G.~Fainekos.
\newblock Linear hybrid system falsification with descent.
\newblock Technical Report arXiv:1105.1733, 2011.

\bibitem{Yashwant:11}
Y.~Annpureddy, C.~Liu, G.~Fainekos, and S.~Sankaranarayanan.
\newblock {S}-{T}a{L}i{R}o: A tool for temporal logic falsification for hybrid
  systems.
\newblock In P.~A. Abdulla and K.~M. Leino, editors, {\em TACAS}, volume 6605
  of {\em LNCS}, pages 254--257. Springer Berlin Heidelberg, 2011.

\bibitem{Ascher:95}
U.~M. Ascher, R.~M.~M. Mattheij, and R.~D. Russell.
\newblock {\em Numerical Solution of Boundary Value Problems for Ordinary
  Differential Equations}.
\newblock SIAM, 1995.

\bibitem{Bertsekas:98}
D.~P. Bertsekas.
\newblock {\em Network optimization: continuous and discrete models}.
\newblock Athena Scientific Belmont, 1998.

\bibitem{Betts:98}
J.~T. Betts.
\newblock Survey of numerical methods for trajectory optimization.
\newblock {\em Journal of Guidance, Control, and Dynamics}, 21(2), 1998.

\bibitem{Branicky:98a}
M.~S. Branicky, V.~S. Borkar, and S.~K. Mitter.
\newblock A unified framework for hybrid control: Model and optimal control
  theory.
\newblock {\em IEEE Transactions on Automatic Control}, 43(1):31--45, 1998.

\bibitem{Branicky:06}
M.~S. Branicky, M.~M. Curtiss, J.~Levine, and S.~Morgan.
\newblock Sampling-based planning, control and verification of hybrid systems.
\newblock {\em IEE Proceedings-Control Theory and Applications},
  153(5):575--590, 2006.

\bibitem{Dang:09}
T.~Dang and T.~Nahhal.
\newblock Coverage-guided test generation for continuous and hybrid systems.
\newblock {\em Formal Methods in System Design}, 34(2):183--213, 2009.

\bibitem{Dijkstra:59}
E.~Dijkstra.
\newblock A note on two problems in connexion with graphs.
\newblock {\em Numerische Mathematik}, 1(1):269--271, 1959.

\bibitem{Dzetkulic:11}
T.~Dzetkuli{\v c} and S.~Ratschan.
\newblock Incremental computation of succinct abstractions for hybrid systems.
\newblock In {\em FORMATS 2011}, volume 6919 of {\em LNCS}, pages 271--285.
  Springer, Heidelberg (2011), 2011.

\bibitem{Edelkamp:09}
S.~Edelkamp, V.~Schuppan, D.~Bošnački, A.~Wijs, A.~Fehnker, and H.~Aljazzar.
\newblock Survey on directed model checking.
\newblock In D.~Peled and M.~Wooldridge, editors, {\em Model Checking and
  Artificial Intelligence}, volume 5348 of {\em Lecture Notes in Computer
  Science}, pages 65--89. Springer Berlin Heidelberg, 2009.

\bibitem{Fehnker:04}
A.~Fehnker and F.~Ivan{\v c}i{\'c}.
\newblock Benchmarks for hybrid systems verification.
\newblock In R.~Alur and G.~J. Pappas, editors, {\em HSCC'04}, number 2993 in
  LNCS. Springer, 2004.

\bibitem{Gendreau:10}
M.~Gendreau and J.-Y. Potvin, editors.
\newblock {\em Handbook of Metaheuristics}.
\newblock Springer, 2nd edition, 2010.

\bibitem{Hiskens:00}
I.~Hiskens and M.~Pai.
\newblock Trajectory sensitivity analysis of hybrid systems.
\newblock {\em IEEE Transactions on Circuits and Systems I: Fundamental Theory
  and Applications}, 47(2):204--220, Feb 2000.

\bibitem{Lamiraux:04}
F.~Lamiraux, E.~Ferr{\'e}, and E.~Vall{\'e}e.
\newblock Kinodynamic motion planning: connecting exploration trees using
  trajectory optimization methods.
\newblock In {\em IEEE International Conference on Robotics and Automation},
  volume~4, pages 3987--3992. IEEE, 2004.

\bibitem{LaValle:06}
S.~M. LaValle.
\newblock {\em Planning Algorithms}.
\newblock Cambridge University Press, 2006.

\bibitem{Locatelli:13}
M.~Locatelli and F.~Schoen.
\newblock {\em Global Optimization--Theory, Algorithms, and Applications}.
\newblock SIAM, 2013.

\bibitem{Mostermann:99}
P.~J. Mosterman.
\newblock An overview of hybrid simulation phenomena and their support by
  simulation packages.
\newblock In F.~Vaandrager and J.~van Schuppen, editors, {\em HSCC'99}, number
  1569 in LNCS. Springer, 1999.

\bibitem{Nghiem:10}
T.~Nghiem, S.~Sankaranarayanan, G.~Fainekos, F.~Ivan{\v c}i\'{c}, A.~Gupta, and
  G.~J. Pappas.
\newblock Monte-carlo techniques for falsification of temporal properties of
  non-linear hybrid systems.
\newblock In {\em HSCC '10}, pages 211--220. ACM, 2010.

\bibitem{Plaku:09}
E.~Plaku, L.~E. Kavraki, and M.~Y. Vardi.
\newblock Falsification of {LTL} safety properties in hybrid systems.
\newblock In {\em Tools and Algorithms for the Construction and Analysis of
  Systems}, volume 5505 of {\em LNCS}, pages 368--382, 2009.

\bibitem{Plaku:09a}
E.~Plaku, L.~E. Kavraki, and M.~Y. Vardi.
\newblock Hybrid systems: from verification to falsification by combining
  motion planning and discrete search.
\newblock {\em Formal Methods in System Design}, {34}({2}):{157--182}, {APR}
  {2009}.

\bibitem{Pohl:71}
I.~Pohl.
\newblock Bi-directional search.
\newblock {\em Machine Intelligence}, 6:124--140, 1971.

\bibitem{Ratschan:09a}
S.~Ratschan and J.-G. Smaus.
\newblock Finding errors of hybrid systems by optimising an abstraction-based
  quality estimate.
\newblock In C.~Dubois, editor, {\em Tests and Proofs}, volume 5668 of {\em
  LNCS}, pages 153--168. Springer, 2009.

\bibitem{Rios:12}
L.~M. Rios and N.~V. Sahinidis.
\newblock Derivative-free optimization: A review of algorithms and comparison
  of software implementations.
\newblock {\em Journal of Global Optimization}, pages 1--47, 2012.

\bibitem{Schoen:02}
F.~Schoen.
\newblock Two-phase methods for global optimization.
\newblock In P.~Pardalos and H.~Romeijn, editors, {\em Handbook of Global
  Optimization}, volume~62 of {\em Nonconvex Optimization and Its
  Applications}, pages 151--177. Springer US, 2002.

\bibitem{Zutshi:13}
A.~Zutshi, S.~Sankaranarayanan, J.~V. Deshmukh, and J.~Kapinski.
\newblock A trajectory splicing approach to concretizing counterexamples for
  hybrid systems.
\newblock In {\em CDC'13}, 2013.

\end{thebibliography}

\AL{
\appendix
\section{Appendix}

Proof of Lemma~\ref{lem:SSiM}:
\begin{proof}
The trivial case is when $r$ does not share a mode with any $p_i$, $i = 1, \ldots, n$. Then it follows from Definition~\ref{def:dist_states} that $d(r, p_i) = d(p_i, r) = \infty$ for all $i = 1, \ldots, n$. Since we assume $c(p_1, \ldots, p_n)$ is finite we get $d_I(r) = d_U(r) = \infty$ as well. Therefore, such a new state $r$ cannot improve the value of $c(p_1, \ldots, p_n)$.

In order to improve on the value of $c(p_1, \ldots, p_n)$ state $r$ must share a mode with some $p_i$'s in $(p_1, \ldots, p_n)$. These $p_i$'s form a subpath $(p_1^s, \ldots, p_k^s)$ and suppose states $p_i^s$ do not belong to modes where $\Init$ and $\Unsafe$ are. 

 First, assume we have a new path $(p_1^s, \ldots, p_{i}^s, r, p_{i+1}^s, \ldots, p_k^s)$. Then, it is clear that because of the triangle inequality $c(p_1^s, \ldots, p_{i}^s, r, p_{i+1}^s, \ldots, p_k^s) \geq c(p_1^s, \ldots, p_k^s)$. Hence we may want to substitute state $r$ for one or more states $p_i^s$. This is a more complicated, however, the same argument with the triangle inequality together with the minimality of $c(p_1, \ldots, p_n)$ apply. Note that, when we use $r$ instead of some $p_i^s$, we may remove a state which is in relation given by $\rightarrow$. If this is a case, the resulting cost gets larger, even infinite, and not lower as we intended.

Now, assume state $r$ is in a mode where either $\Init$ or $\Unsafe$ are. Since, this turns to be a symmetrical problem, we shall investigate only the case when $r$ and $\Init$ share the same mode. If $r \in \Init$, then we are done, since $d_I(p_1) = \min_{u \in \Init}d(u, p_1)$.  On the other hand, if $r \notin \Init$, we use the following argument. Since $\Init$ is closed, bounded and convex, there exists an unique element $p_1^\star \in \Init$ such that
\AL{\[
d(p_1^\star, p_1) = \min_{u \in \Init}d(u, p_1).
\]}
{$
d(p_1^\star, p_1) = \min_{u \in \Init}d(u, p_1).
$. }
Again, using the triangle inequality $d(p_1^\star, p_1) \leq d(p_1^\star, r) + d(r, p_1)$. In addition, according to the argument above, we cannot use $r$ instead of any state from path $(p_1, \ldots, p_n)$. \qed
\end{proof}

Proof of Theorem~\ref{thm:1}:
\begin{proof}
Let $E$ be the error trajectory, and let $L$ be its length wrt. time. Without loss of generality we assume that $L\geq T$.

We will work with a sequence of hybrid trajectories with $\Ceil{\frac{L}{T}}$ segments. Observe that---according to Definition~\ref{def:cost}---the allowed cost $\varepsilon$ from which convergence of local search is guaranteed, consists of $\Ceil{\frac{L}{T}}+1$ summands ($\Ceil{\frac{L}{T}}-1$ times the cost for joining segments, the cost to initial, and the cost to unsafe). We will distribute the allowed cost $\varepsilon$ equally among those, allowing error $\varepsilon':=\frac{\varepsilon}{\Ceil{\frac{L}{T}}}+1$ for each summand.  

We will now analyze the summands of the cost:
\begin{itemize}
\item The cost resulting from the starting point of the simulation: With probability $1$ the algorithm will eventually produce a segment 
  \begin{itemize}
  \item the starting point of the segment is  closer than $\varepsilon'$ from the starting point of $E$, and for which 
  \item all other points have distance less than $\frac{\varepsilon'}{2}$ from a point on $E$ (due to continuous dependency of trajectories on the initial values). 
  \end{itemize}
\item The cost between the end-point of any segment and starting point of the next segment: We first consider the first segment: Let $p$ be the point on $E$ that is closest to the end-point of the first segment. This distance is smaller than $\frac{\varepsilon'}{2}$ since the algorithm chose the starting point of the segment close enough to the starting point of $E$. With probability $1$ the algorithm will eventually start another simulation such that 
  \begin{itemize}
  \item the starting point of the segment is  closer than $\frac{\varepsilon'}{2}$ from $p$, and for which 
  \item all other points have distance less than $\frac{\varepsilon'}{2}$ from a point on $E$ (due to continuous dependency of trajectories on the initial values). 
  \end{itemize}
   We can continue with the same argument along $E$ until the final segment. 
\item The cost between the end-point of the final segment and the set of unsafe states: Let $p$ be the point on $E$ that is closest to the end-point of the last but one segment, as in the previous item. Since the final point of $E$ is in the interior of the set of unsafe points, there is a $\mu>0$ such that all points of distance smaller than $\mu$ from this final point is also unsafe.  Again, with probability $1$ we will eventually start another simulation whose distance from $E$ is smaller than $\mu$ which means that this simulation will hit an unsafe state.
\end{itemize}

The resulting path has cost not bigger than $\varepsilon$. Hence local search from this path will converge to an error trajectory. The path of minimal cost used by the algorithm might be a different one that this one. However, that path of minimal cost will again have cost not bigger than $\varepsilon$ which guarantees convergence. \qed
\end{proof}
}
{}

\comment{
\appendix

\section{Questions/Problems}

\begin{itemize}
\item How to learn from failing local optimization? find smallest sub-path that
  also fails, and learn something about it 
\item local optimization over modes (i.e., changing sequences of modes)
\item when and how should optimization change $P$, which initial $P$
\item adapt/improve optimal control algorithms for our purpose (e.g., stopping criteria)
\item how to exploit linearity? (see LQR trees)
\item exploit verification tools
\item Optimal anchor point for shooting methods? shooting that is
  forward/backward symmetric?
\end{itemize}
}

\end{document}